\documentclass[11pt]{article}

\usepackage[english]{babel}
\usepackage[utf8x]{inputenc}
\usepackage{amsmath}
\usepackage{amssymb}

\usepackage{hyperref}

\usepackage{amsthm}
\usepackage{empheq}
\usepackage{graphicx}
\usepackage[colorinlistoftodos]{todonotes}
\usepackage[left=1in,right=1in,top=1in,bottom=1in,bindingoffset=0cm]{geometry}
\def\P{\mathbb{P}}

\def\E{\mathbb{E}}

\def\R{\mathbb{R}}

\def\e{\varepsilon}

\def\e{\varepsilon}
\def\g{\gamma}

\newtheorem{claim}{Claim}
\newtheorem{thm}{Theorem}
\newtheorem{lem}{Lemma}

\newtheorem{conj}{Conjecture}
\newcommand{\korner}{K\"{o}rner}
\newcommand{\komlos}{Koml\'{o}s}

\title{Beating Fredman-Koml\'{o}s for perfect $k$-hashing\thanks{Research supported in part by NSF CCF-1422045 and CCF-1563742.}}

\author{ Venkatesan Guruswami\thanks{Some of this work was done when the author was visiting the School of Physical and Mathematical Sciences, Nanyang Technological University, Singapore and the Center of Mathematical Sciences and Applications, Harvard University.} \and Andrii Riazanov }

\date{Computer Science Department \\ Carnegie Mellon University \\ Pittsburgh, PA 15213. \\ ~ \\ \{venkatg,riazanov\}@cs.cmu.edu}

\begin{document}

\maketitle
\thispagestyle{empty}

\begin{abstract}
We say a subset $C \subseteq \{1,2,\dots,k\}^n$ is a $k$-hash code (also called $k$-separated) if for every subset of $k$ codewords from $C$, there exists a coordinate where all these codewords have distinct values. Understanding the largest possible rate (in bits), defined as $(\log_2 |C|)/n$, of a $k$-hash code is a classical problem. It arises in two equivalent contexts: (i)  the smallest size possible for a perfect hash family that maps a universe of $N$ elements into $\{1,2,\dots,k\}$, and (ii) the zero-error capacity for decoding with lists of size less than $k$ for a certain combinatorial channel.

\smallskip
A general upper bound of $k!/k^{k-1}$ on the rate of a $k$-hash code (in the limit of large $n$) was obtained by Fredman and Koml\'{o}s in 1984 for any $k \geq 4$. While better bounds have been obtained for $k=4$, their original bound has remained the best known for each $k \ge 5$.
In this work, we obtain the first improvement to the Fredman-Koml\'{o}s bound for every $k \ge 5$. While we get explicit (numerical) bounds for $k=5,6$, for larger $k$ we only show that the FK bound can be improved by a positive, but unspecified, amount.  Under a conjecture on the optimum value of a certain polynomial optimization problem over the simplex, our methods allow an effective bound to be computed for every $k$.
\end{abstract}

\newpage

\section{Introduction}

A code of length $n$ over an alphabet of size $k$ is a subset $C \subseteq \{1, 2,\dots, k\}^n$. We say such a code $C$ is a \emph{$k$-hash code} (also called $k$-separated in the literature), if for every subset of $k$ distinct codewords $\{c^{(1)}, c^{(2)}, \dots, c^{(k)}\}$ from $C$, there exists a coordinate $j$ such that all these codewords differ in this coordinate, i.e. $\{c^{(1)}_j, c^{(2)}_j, \dots, c^{(k)}_j\} = \{1, 2,\dots, k\}$. The rate (in bits) of the code is defined as $R = \frac{\log_2 |C|}{n}$. Then for each fixed integer $k$, let $R_k$ be the supremum, as $n \to \infty$, of the rate of the largest $k$-hash code of length $n$. 

The study of the quantity $R_k$ is a fundamental problem in combinatorics, information theory, and computer science. As the name suggests, $k$-hash codes have strong connections to the hashing problem. A family of functions mapping a universe of size $N$ to the set $\{1, 2, \dots, k\}$ is called a perfect $k$-hash family if any $k$ elements of the universe are mapped in one-to-one fashion by at least one hash function from this family. If $C$ is a $k$-hash code, then a perfect $k$-hash family for universe $C$ with $n$ functions is just the family of coordinate projections. Therefore, $R_k$ gives the growth rate of the size of universes for which perfect $k$-hash families of a given size exist. Equivalently, an upper bound on $R_k$ is equivalent to a lower bound on the size of a perfect $k$-hash family as a function of the universe size.

An equivalent information-theoretic context in which $k$-hash codes arise concerns zero-error list decoding on certain channels. A channel can be thought of as a bipartite graph $(V, W, E)$, where $V$ is the set of  channel inputs, $W$ is the set of channel outputs, and $(v, w) \in E$ if on input $v$ the channel can output $w$. The $k/(k-1)$ channel then is the channel with $V = W = \{1, 2, \dots, k\}$, and $(v, w) \in E$ iff $v \not= w$. In this context, $R_k$ is the largest asymptotic rate at which one can communicate using $n$ repeated uses of the channel (as $n$ grows), if we want to ensure that the receiver can identify a subset of at most $k-1$ sequences that is guaranteed to contain the transmitted sequence. See \cite{elias88,Dalai} for more details. 

Studying the rates of the codes and hashing family sizes in the above settings is a longstanding problem.  A probabilistic argument shows the existence of $k$-hash codes with rate at least $\frac{1}{k-1}\log \frac{1}{1- k!/k^k} - o(1)$~\cite{Fredman,Korner86}, and better bounds are known for some small values of $k$. Our focus here is on \emph{upper bounds} on $R_k$, that is limitations on the size of $k$-hash codes. Here the best-known general upper bound on the rate $R_k$ dates all the way back to the 1984 paper of Fredman and \komlos~\cite{Fredman}:
\begin{equation}
\label{intro-Fredman}
R_{k} \leq \dfrac{k!}{k^{k-1}} =: \alpha_{k} \ . 
\end{equation}
For large $k$ the multiplicative discrepancy between the probabilistic lower bound on $R_k$ and the above Fredman-\komlos\ upper bound \eqref{intro-Fredman} grows approximately as $k^2$, so the current bounds on the rate require tightening to obtain better estimations of $R_k$.
There is another trivial upper bound, ${R_k\le\log_2\left( \frac{k}{k-1} \right)}$, that follows from a simple double-counting or first moment method. The above bound \eqref{intro-Fredman} is much better than this bound for $k \ge 4$. For $k=3$ (which is called the trifference problem by K\"{o}rner), however, $R_3 \le \log_2(3/2) \approx 0.585$ remains the best upper bound, and improving it (or showing it can be achieved!) is a major combinatorial challenge. 
For the case $k=4$, the bound \eqref{intro-Fredman} which states $R_4 \le 0.375$ has been improved, first by Arikan to $0.3512$~\cite{Arikan}, and recently by Dalai, Guruswami, and Radhakrishnan~\cite{Dalai} to $6/19 \le 0.3158$. 

However, the above quantity $\alpha_k$ remained the best known upper bound on $R_k$ for each $k > 4$. Our main result gives the first improvement to the Fredman-\komlos\ bound \eqref{intro-Fredman} for $k \ge 5$, proving that $R_k$ is strictly smaller than $\alpha_k$ for every $k$. 
\begin{thm}
For all $k \ge 4$ there exists $\beta_k$ such that $R_k \leq \beta_k < \alpha_k$. For $k=5,6$, we have the explicit upper bounds $R_5 < \beta_5 = 0.190825 < 0.192 = \alpha_5$, and $R_6 < \beta_6 = 0.0922787 < 0.0\overline{925} = \frac{5}{54} =\alpha_6$. 
\end{thm}
Our proof relies on the continuity arguments, which doesn't yield an effective way to find explicit values of $\beta_k$. 
However, we present a technical approach to compute our upper bound $\beta_k$, which relies on polynomial optimization over the simplex. Using the tools of numerical optimization in {\it Mathematica}, we calculate the values of $\beta_5$ and $\beta_6$. We also make a conjecture on the optimum value of a certain polynomial optimization problem over the simplex, assuming which our methods allow an effective upper bound that is strictly smaller than $\alpha_k$ to be computed for every $k$. 

 Our approach is also applicable to the $(b,k)$-hashing problem for $b \ge k$, where one considers codes $C\subseteq \{1, 2,\dots, b\}^n$ with the property that for any $k$ distinct codewords $\{c^{(1)}, c^{(2)}, \dots, c^{(k)}\}$ from $C$ there exists a coordinate $j$ such that all these codewords differ in this coordinate. Using exactly the same arguments, we obtain an improvement on the K\"{o}rner-Marton upper bound~\cite{Korner} on the rate of such codes. When $b=k$, this latter bound is identical to the Fredman-\komlos\ bound, but can be better than the corresponding bound in \cite{Fredman} when $b > k$.
 For some small pairs of values $(b,k)$, with $b > k$, the K\"{o}rner-Marton bound was further improved by Arikan~\cite{Arikan}. In those cases, the bounds we get are probably weaker than Arikan's. For this reason and for the sake of simplicity, in this paper we analyze only the case $b = k$, which corresponds to $k$-hashing. But all our proofs generalize in a straightforward way for $(b,k)$-hashing as well. We briefly describe the background on the $(b,k)$-hashing problem in  Appendix \ref{(b,k)-sec}.

\section{Background and approach}
\label{approach}

The previous general upper bounds on the rates of $k$-hash codes by Fredman and Koml{\'o}s \cite{Fredman}, \korner\ and Marton \cite{Korner}, and Arikan \cite{Arikan} are all based on information-theoretic inequalities for graph covering, related to the Hansel lemma~\cite{hansel}. \korner~\cite{Korner86} cast the Fredman-\komlos\ proof in the language of graph entropy, which he had introduced in \cite{Kornerentropy} (see \cite{Radhakrishnan} for a nice survey on graph entropy). \korner\ and Marton~\cite{Korner} generalized this approach to the hypergraph case, which led to improvements to the Fredman-\komlos\ bound for the $(b,k)$-hashing problem in certain cases when $b > k$, but not for $R_k$. 
In this paper we use the following version of the Hansel lemma, which is also proved in \cite{Nilli} via a simple probabilistic argument:
\begin{lem}[Hansel]
\label{Hansel-lemma}
Let $K_m$ be a complete graph on $m$ vertices. Let also $G_1, G_2, \dots, G_t$ be bipartite graphs, such that $E\left(K_m\right) = \bigcup\limits_{i=1}^tE(G_i)$. Denote by $\tau(G_i)$ the fraction of non-isolated vertices in $G_i$. Then the following holds:
\begin{equation}
\label{Hansel}
\log m \leq \sum_{i=1}^t\tau(G_i).
\end{equation}
\end{lem}

To relate this lemma to the context of the paper, consider a $k$-hash code $C\subseteq [k]^n$. Take a subset of this code $\{x_1, x_2, \dots, x_{k-2}\}\subseteq C$, and define bipartite graphs $G_i^{x_1, \dots, x_{k-2}}$, for $i \in [n]$, as follows: 
\begin{equation*}
\begin{aligned}
& V(G_i^{x_1, \dots, x_{k-2}}) = C\setminus \{x_1, x_2, \dots, x_{k-2}\},\\
& E(G_i^{x_1, \dots, x_{k-2}}) = \Bigg\{ \{y_1, y_2\}\ :\ (y_{1})_i, (y_2)_i, (x_1)_i, (x_2)_i, \dots, (x_{k-2})_i \text{ are distinct} \Bigg\}.
\end{aligned}
\end{equation*}

Note that since $C$ is a $k$-hash code, for any pair $\{y_1, y_2\}\subseteq C\setminus \{x_1, x_2, \dots, x_{k-2}\}$, there exists some coordinate $i$, such that all the $k$ codewords $y_1, y_2, x_1, x_2, \dots, x_{k-2}$ differ in the $i^{\text{th}}$ coordinate. In other words, $\{y_1, y_2\} \in E(G_i^{x_1, \dots, x_{k-2}})$ for this $i$. Therefore, it holds $E\left(K_{|C|-(k-2)}\right) = \bigcup\limits_{i=1}^nE(G_i^{x_1, \dots, x_{k-2}})$. Then Hansel lemma \ref{Hansel-lemma} applies directly, and denoting $\tau_i(x_1, x_2, \dots, x_{k-2}) = \tau\left(G_i^{x_1, \dots, x_{k-2}}\right)$, we obtain
\begin{equation}
\label{Hansel-code}
\log \left(|C| - k + 2\right) \leq \sum_{i=1}^n\tau_i(x_1, x_2, \dots, x_{k-2}).
\end{equation}      
Taking the expectation over the choice of $x_1, x_2, \dots, x_{k-2}$, we get

\begin{equation}
\label{EHansel-code}
\log \left(|C| - k + 2\right) \leq \sum_{i=1}^n\E[\tau_i(x_1, x_2, \dots, x_{k-2})].
\end{equation}  

By bounding the RHS of the above inequality one might obtain an upper bound on $\log |C|$, and thus on the rate of this code. Different strategies to pick the codewords $\{x_1, x_2, \dots, x_{k-2}\}$ from $C$ lead to different approaches to bound the RHS of \eqref{EHansel-code}. Here we briefly present the ideas underlying the previous works and then outline our approach.

	In the original bound by Fredman and \komlos~\cite{Fredman}, the codewords $x_1, x_2, \dots, x_{k-2}$ are picked independently at random from the code $C$. Then one can use symmetry arguments (or Muirhead's inequality) to give an upper bound on the RHS of \eqref{EHansel-code}, which will lead to the bound on the rate
	\begin{equation}
	\label{Fredman-Kolmos} R_{k} \leq \dfrac{k!}{k^{k-1}}.
	\end{equation}
	Due to the symmetry arguments involved, this bound is actually tight only in the case when the frequencies of the symbols of  the code $C$ in each coordinate are uniform. 

\medskip
Arikan \cite{Arikan,Arikan2} used rate versus distance results from the coding theory to ensure that it is possible to pick $x_1, x_2, \dots, x_{k-2}$ which agree on many coordinates. Note that this already guarantees that many terms in the RHS of \eqref{Hansel-code} equal $0$. Together with an argument which allows to modify the code so that it doesn't have any coordinate where the symbols have an overly skewed frequency, Arikan was able to improve the bound \eqref{Fredman-Kolmos} for $k=4$. However, no improvement was gained for larger $k$.

\medskip 
Dalai, Guruswami, and Radhakrishnan \cite{Dalai} combine aspects of the above two approaches for the case $k = 4$. As in Arikan's work, they pick $x_1,x_2$ to agree on the first several coordinates. However, instead of a fixed such choice, they pick such a pair at random (from a natural distribution). The technical crux of their argument is a concavity claim for some quadratic form which says that despite conditioning on a common prefix, which might greatly alter the frequency vector of symbols in any coordinate in the suffix, the Fredman-\komlos\ bound for completely random $x_1,x_2$ is still valid on those coordinates. (In some sense, only the average frequency vector over all prefixes matters, not the individual ones.) Actually this holds modulo a technical condition that there are no coordinates with very skewed symbol distribution, which can be ensured by some pre-processing of the code similar to \cite{Arikan2}. Thus some terms in \eqref{EHansel-code} are equal to $0$ and the other are bounded by $3/8$, and balancing these appropriately, a bound of $R_4 \le \frac{6}{19}$ is obtained in \cite{Dalai}.

\bigskip In this work, we follow the strategy of \cite{Dalai} for general $k$ by picking $x_1, x_2, \dots, x_{k-2}$ randomly so that they all lie in a subcode of $C$ that takes at most $(k-3)$ values on each coordinate from the the large set $T$. This again implies that the coordinates from $T$ contribute $0$ to the RHS of \eqref{EHansel-code}. In this case, however, the analogous concavity claim seems out of reach, as one has to argue about degree $(k-2)$ polynomials rather than quadratics. 
 We instead take a different approach that works directly with the arbitrary symbol frequencies that may arise upon conditioning within a subcode, avoiding the averaging or concavity step. (This leads to worse bounds, but our main goal is to beat the Fredman-\komlos\ bound for $k > 4$ by some positive amount at all, doesn't matter how small.)  However, another problem arises in that the constraint on the code to have non-skewed frequencies in each coordinate cannot be dealt with using Arikan's argument for large $k$. To cope with this issue, we differentiate two separate cases: (i) where $C$ has only a few coordinates with skewed distributions of symbols, and (ii) where there are a lot of such coordinates. 
 \begin{itemize}
 \item In the first case, we pick the coordinates $T$ (where $x_1,x_2,\dots,x_{k-2}$ are chosen to collide) to include all these skewed coordinates.  Note that this is unlike \cite{Arikan,Dalai} where any choice of $T$ of prescribed size works. 
 Our choice of $T$ ensures that in the remaining coordinates the frequency vector is not too far from uniform, and we apply the approach of \cite{Dalai} directly. We reduce the task of showing that this improves upon the Fredman-\komlos\ bound to the continuity of certain functions, which we argue via Berge's maximum principle (\cite{Efe}).
 \item In the second case, we use the original random strategy of picking $x_1, x_2, \dots, x_{k-2}$ as in \cite{Fredman}. The idea here is that the bound \eqref{Fredman-Kolmos} is tight only when all the frequencies of symbols are exactly uniform. Then, in the case when there are a lot of far-from-uniform frequencies, it is possible to improve the bound \eqref{Fredman-Kolmos}. 
 \end{itemize}
 By picking the correct way to differentiate between skewed and non-skewed distributions, we then obtain an improvement on the Fredman-\komlos\ bound \eqref{Fredman-Kolmos} for every $k \geq 4$. As mentioned earlier, this is the \emph{first} such improvement for $k \ge 5$. For $k = 5$ and $k=6$ we use numerical optimization tools to provide new explicit bounds on $R_k$, presented in section \ref{explicit_bounds}.

\section{Upper bound on the rate of $k$-hash codes}
Let $\Sigma = \{1, 2, \dots, k\} = [k]$, and let $C \subseteq \Sigma^n$ be a $k$-hash code with rate $R = \frac{\log |C|}{n}$ (all logarithm are to the base $2$). Let $f_i \in \R^k$ be the frequency vector of symbols of the code for each coordinate $i\in [n]$, namely:
\[ f_i[a] = \dfrac{1}{|C|}|\{x\in C : x_i = a\}|. \]

Throughout the analysis, we will be interested in two cases: when for most of the coordinates the distribution of codeword symbols is close to uniform (non-skewed), or when this doesn't hold. To define the term ``close to uniform'' formally, we consider a \emph{threshold} $\g$, that satisfies $\frac1{2k-3} \leq \g \leq \frac1k$, and say that $f\in \R^k$ is close to uniform when for all components of $f$ it holds $f[a] \geq \g$. Denote then $P_{\gamma} = \{ i\in [n] : \min\limits_{a\in \Sigma} f_i[a] \geq \g \}$ -- the set of all the coordinates for which the distribution of codeword symbols is close to uniform. Denote also $\ell = \left\lfloor \frac{nR - \log n }{\log \left(\frac{k}{k-3}\right)}\right\rfloor$. We then consider two cases:
\begin{enumerate}
\item {\it Unbalanced:} $|P_{\g}| < n - \ell$, so there is a decent fraction of coordinates where the distribution of codeword symbols is skewed. For this case, we apply the random strategy to pick $x_1, x_2,\dots, x_{k-2}$ in \eqref{EHansel-code}.
\item {\it Almost balanced:} $|P_{\g}| \geq n - \ell$, so for almost all coordinates, the distribution of codeword symbols is close to uniform. Then we follow the approach from \cite{Dalai} to pick $x_1, x_2,\dots, x_{k-2}$ which collide on many coordinates. 
\end{enumerate}

For both of these cases, we will obtain some bounds on the rate of $C$, which depend on the threshold $\g$. It then will remain to choose $\g$ in a manner ensuring that both these bounds beat \eqref{Fredman-Kolmos}. Then, since for any code $C$ exactly one of the cases holds, we can obtain a general upper bound on the rate. 

Before we continue with studying the two cases separately, let's look at how we can estimate $\tau_i(x_1, x_2, \dots, x_{k-2})$. Clearly, the codeword $y \in C$ appears non-isolated in the graph $G_i^{x_1, x_2, \dots, x_{k-2}}$ if and only if all the codewords $x_1, x_2, \dots, x_{k-2}$ and $y$ differ is the $i^{\text{th}}$ coordinate. Therefore, the fraction of non-isolated vertices in $G_i^{x_1, x_2, \dots, x_{k-2}}$ is exactly
\begin{align}
\label{tau_precise}\nonumber \tau_i(x_1, &\dots, x_{k-2}) \\&= \left(\dfrac{|C|}{|C|-(k-2)}\right)\bigg(1 - f_i[x_{1i}] - f_i[x_{2i}] - \dots - f_i[x_{(k-2)i}]\bigg) \mathbf{1}[x_{1i}, x_{2i}, \dots, x_{(k-2)i} \text{ distinct}], 
\end{align}
where $\mathbf{1}(E)$ is the indicator variable for an event/condition $E$.

\subsection{Unbalanced case}
\label{unbalanced_sec}
We will pick $x_1, x_2, \dots, x_{k-2}$ uniformly at random without replacement from $C$ to obtain an upper bound on the rate of $C$ from \eqref{EHansel-code}. Taking the expectations of the both sides in \eqref{tau_precise}, we obtain

\begin{equation}
\begin{aligned}
\label{tau_exp} \E[\tau_i&(x_1, \dots, x_{k-2})] \\
&= \dfrac{|C|}{|C|-k + 2} \sum_{\substack{a_1,\dots, a_{k-2}\ \in\ \Sigma\\ \{a_s\} \text{ distinct}}} \left(1 - \sum_{s=1}^{k-2}f_i[a_s] \right)\cdot\P\big[ (x_s)_i = a_s,\ s = 1, \dots, (k-2) \big] \\  
 &= \dfrac{|C|}{|C|-k+2}\dfrac{|C|}{|C|}\dfrac{|C|}{|C|-1}\dots \dfrac{|C|}{|C|-(k-3)} \sum_{\substack{a_1,a_2,\dots, a_{k-2}\ \in\ \Sigma\\ \{a_s\} \text{ distinct}}} \left(1 - \sum_{s=1}^{k-2}f_i[a_s] \right)\cdot f_i[a_1]f_i[a_2]\dots f_i[a_{k-2}], 
 \end{aligned}
 \end{equation}
where the coefficients $\frac{|C|}{|C| - j}$, $j = 0, 1, \dots, k-3$ appear because we pick elements from $C$ without replacement. We then define the following function of two probability vectors $g, f \in \R^k$:
\begin{equation}
\label{def-phi} \phi_k(g, f) = \sum_{\substack{a_1,a_2,\dots, a_{k-2}\in\Sigma\\ \{a_s\} \text{ distinct}}}\ \prod_{s=1}^{k-2}g[a_s]\bigg(1 - \sum_{s=1}^{k-2}f[a_s]\bigg). 
\end{equation}
Then the inequality \eqref{tau_exp} can be written as follows:
\begin{equation}
\label{unbal_tau_phi}
\begin{aligned}
\E[\tau_i(x_1, x_2, \dots, x_{k-2})] \leq \phi_k(f_i, f_i)\big(1 + o(1)\big).
\end{aligned}
\end{equation}

Since $\sum_{a\in \Sigma}f_i[a] = 1$, it is easy to see that $\phi_k(f_i, f_i)$ is a symmetric expression in $f_i[a]$ for all $a\in \Sigma$. Denote by $S_h^t(g)$ the $h$-th elementary symmetric sum of the first $t$ coordinates of the vector $g\in \R^k$, i.e. the sum of all products of $h$ distinct elements from $\{g[1], g[2], \dots, g[t]\}$. For example,
\[ S_3^4(g) = g[1]g[2]g[3] + g[1]g[2]g[4] + g[1]g[3]g[4] + g[2]g[3]g[4].  \]
Then we can write 
\[ \phi_k(f_i, f_i) = (k-2)!\cdot\binom{k-1}{k-2}S^k_{k-1}(f_i) = (k-1)!\cdot S^k_{k-1}(f_i)\]

It is easy to see that $S_h^k(g)$ for $g$ being a probability vector in $\R^k$ is maximized when $g$ is uniform. Indeed, if there are two non-equal coordinates $g[a] \not= g[b]$, then substituting the values in these coordinates by their arithmetic average strictly increases the value of $S_h^k(g)$. Then let us denote by $u$ the uniform distribution on $k$ elements, i.e. $u[a] = 1/k$ for all $a \in [k]$, and so it holds $S^k_{h}(g) \leq S^k_h(u)$. Then in \eqref{unbal_tau_phi} we obtain
\[ \E[\tau_i(x_1, x_2, \dots, x_{k-2})] \leq (k-1)!\cdot S^k_{k-1}(f_i)\cdot (1 + o(1)\big) \leq (k-1)!\cdot S^k_{k-1}(u)\cdot \big(1 + o(1)\big),\]
where it holds
\[ S^k_{k-1}(u) = \binom{k}{k-1}\cdot\left(\frac1k\right)^{k-1} = \dfrac1{k^{k-2}}.\]
Therefore, we retrieve
\begin{equation}
\label{unbal-uniform} \E[\tau_i(x_1, x_2, \dots, x_{k-2})] \leq  \dfrac{(k-1)!}{k^{k-2}}\cdot \big(1 + o(1)\big) = \dfrac{k!}{k^{k-1}}\cdot \big(1 + o(1)\big).
\end{equation}
Substituting this inequality into \eqref{EHansel-code}, notice that we derive exactly the Fredman-Koml{\'o}s bound \eqref{Fredman-Kolmos}. Denote then \[\alpha_{k} = \dfrac{k!}{k^{k-1}},\]
the Fredman-Koml{\'o}s upper bound on the rate $R_k$.

Now recall that we are considering the unbalanced case, in which there are a lot of coordinates with frequencies of codeword symbols being far from uniform. Take $i$ to be any of such coordinates, and let for convenience $f = f_i$, so it holds $\min_{a\in \Sigma}f[a] < \g$. Without loss of generality, say $f[k] < \g$. Notice the following trivial property of symmetric sums:
\[ \phi_k(f, f) = (k-1)!\cdot S^k_{k-1}(f) = (k-1)! \Big(  S^{k-1}_{k-1}(f) + f[k]\cdot S^{k-1}_{k-2}(f)\Big).\]
 
The above expression is symmetric in the first $(k-1)$ coordinates of $f$. Let's then fix $f[k]$, and do the same averaging operations with all the remaining coordinates of $f$, making in the end $f'[1]=f'[2]=\dots=f'[k-1] = \frac{1-f[k]}{k-1}$. The value of $\phi_k(f,f)$ only increases after such operations, so
\[ \phi_k(f,f) \leq (k-1)!\bigg(f'[1]f'[2]\cdots f'[k-1] + f[k]\cdot S_{k-2}^{k-1}(f') \bigg).\]

Denote $y = \frac{1-f[k]}{k-1}$, so $f[k] = 1 - (k-1)y$. Since $0 \leq f[k] < \g$ by the assumption above, it holds $\frac{1-\g}{k-1} \leq y \leq \frac1{k-1}$. Note that we took the threshold $\g \leq \frac1k$, thus $y \geq \frac{1-\g}{k-1} \geq \frac1k$. Then the following inequality holds:
\[ \phi_k(f,f) \leq (k-1)!\bigg(y^{k-1} + \big(1-(k-1)y\big)\cdot(k-1)y^{k-2}\bigg) = (k-1)!y^{k-2}\bigg((k-1) - (k^2-2k)y\bigg).\]
Denote $G_k(y) = (k-1)!y^{k-2}\bigg((k-1) - (k^2-2k)y\bigg)$, so $ \phi_k(f,f) \leq G_k(y)$. Note that \[\left(G_k(y)\right)'= (k-1)!(k-1)(k-2)y^{k-3}\big(1 - ky\big),\] so the derivative of $G_k$ is negative on the interval $\frac1k \leq \frac{1-\g}{k-1} < y \leq \frac1{k-1}$, and it is zero at $y = \frac1k$.  Therefore, we finally obtain for any such $f$:
\begin{equation}
\label{phi_G_k} \phi_k(f,f) \leq \max_{y \in \left[\frac{1-\g}{k-1}, \frac{1}{k}\right]} G_k(y) = G_k\left(\frac{1-\g}{k-1}\right).
\end{equation}
Note that it always holds $G_k\left(\frac{1-\g}{k-1}\right) \leq G_k\left(\frac1k\right) = \alpha_k$ for any $\g \leq \frac1k$, and the strict inequality $G_k\left(\frac{1-\g}{k-1}\right) < G_k\left(\frac1k\right) = \alpha_k$ holds when $\g < \frac1k$.

So if for some coordinate $i$ it holds $\min_{a\in [k]}f_i[a] < \g$, we have the bound
\begin{equation}
\label{unbal-non-uniform} \E[\tau_i(x_1, x_2, \dots, x_{k-2})] \leq  G_k\left(\frac{1-\g}{k-1}\right) \big(1 + o(1)\big).
\end{equation}

For now we obtained two bounds for the summands in the RHS of \eqref{EHansel-code}: 
(i) the bound \eqref{unbal-uniform} holds for all the coordinates, and (ii) the bound \eqref{unbal-non-uniform} holds for the coordinates with codeword symbol frequencies far from uniform. As we noted above, the second bound is strictly stronger then the first bound when we take the threshold $\g < \frac1k$. Also recall that in the unbalanced case which we now consider, there are a lot of coordinates of the second type, so essentially the bound \eqref{unbal-non-uniform} applies many times. Let's now formalize this argument to obtain an improvement on the Fredman-Koml{\'o}s bound for the unbalanced case.

Denote for convenience $\xi_k(\g) = G_k\left(\frac{1-\g}{k-1}\right)$, then we have
\begin{equation}
\label{def_xi_k} \xi_k(\g) = (k-1)!\dfrac{(1-\g)^{k-2}}{(k-1)^{k-2}}\dfrac{(k-1)^2 - (k^2-2k)(1-\g)}{k-1} = \dfrac{(k-2)!(1-\g)^{k-2}\big((k^2-2k)\g + 1\big)}{(k-1)^{k-2}}.
\end{equation} 
Let also $\e_k(\g) = \alpha_k - \xi_k(\g) \geq 0$. Recall that we denoted by $P_{\g}$ the set of coordinates $i$ for which it holds $\min_{a\in \Sigma} f_i[a] \geq \g$. For such $i$ we directly apply the bound \eqref{unbal-uniform}. For all the other coordinates $i\in [n]\setminus P_{\g}$ we use the inequality \eqref{unbal-non-uniform}. In the unbalanced case $|P_{\g}| < n - \ell$, thus $n -|P_{\g}| > \ell$. Applying all these arguments to \eqref{EHansel-code}, we obtain
\begin{equation*}
\begin{aligned}
\log(|C| - k + 2) &\leq \bigg(|P_{\g}|\alpha_k + \big(n-|P_{\g}|\big)\Big(\alpha_k - \e_k(\g)\Big)\bigg)(1 + o(1)) \\ 
&=\bigg(n\alpha_k - \Big(n - |P_{\g}|\Big)\e_k(\g)\bigg)(1 + o(1)) \\
&< \bigg(n\alpha_k - \ell\e_k(\g)\bigg)(1 + o(1)) \\ 
&\leq \left(n\alpha_k - \frac{nR}{\log \left(\frac{k}{k-3}\right)}\e_k(\g) + \frac{\log n}{\log  \left(\frac{k}{k-3}\right)}\e_k(\g) + \e_k(\g)\right)(1 + o(1)) \\ 
&= \left(n\alpha_k - \frac{nR}{\log  \left(\frac{k}{k-3}\right)}\e_k(\g) + o(n)\right)(1 + o(1)),
\end{aligned}
\end{equation*}
where recall $\ell = \left\lfloor \frac{nR - \log n }{\log \left(\frac{k}{k-3}\right)}\right\rfloor$. Since $|C| = 2^{Rn}$ by definition of the rate $R$, the above implies for $n \to \infty$:
\[ R \leq \alpha_k - \frac{R\e_k(\g)}{\log \left(\frac{k}{k-3}\right)} + o(1), \]

\begin{equation}
\label{k-R_unbal}  
\boxed{R^{\text{unbal}}_k(\g) \leq \dfrac{\alpha_k}{1 + \frac{\alpha_k - \xi_k(\g)}{\log \left(\frac{k}{k-3}\right)}} .}
\end{equation}

Note that if we take $\g = \frac1k$ in the above, we obtain $R^{\text{unbal}}_k(1/k) = \alpha_k$, since $\xi_k(1/k) = G\left(\frac{1-1/k}{k-1}\right) = G\left(\frac1k\right) = \alpha_k$.
Now if we take $\g < \frac1k$, we showed above that $\xi_k(\g) < \alpha_k$, and then \eqref{k-R_unbal} will give a better bound on $R^{\text{unbal}}_k(\g)$. In other words,
\begin{equation}
\label{cont_unbal}
 R^{\text{unbal}}_k\left(\dfrac1k - \varepsilon\right) < R^{\text{unbal}}_k\left(\dfrac1k\right) \leq \alpha_k
 \end{equation}
for any small $\varepsilon > 0$. So we beat the Fredman-\komlos\ bound for the unbalanced case for any choice of the threshold $\g < \dfrac1k$.

\subsection{Almost balanced case}
\label{almost-bal-sec}
For this case we follow the approach used in \cite{Dalai} for $4$-hashing. Namely, we will consider a rich subcode of codewords which can take a restricted set of symbols on some fixed set of coordinates, and choose $x_1, x_2, \dots, x_{k-2}$ randomly from this subcode. In the almost balanced case, we are able to ensure that the distributions of codeword symbols in all non-fixed coordinates are close to uniform, which will allow us to use some continuity argument to bound the RHS of \eqref{EHansel-code}.

In this case we assume $|P_{\g}| \geq n - \ell$, so there are at most $\ell$ coordinates where the distribution of codeword symbols is skewed. The set of such coordinates is $\overline{P_{\g}} = [n]\setminus P_{\g}$, $|\overline{P_{\g}}| \leq \ell$. Then take any subset $T \subset [n]$, such that $\overline{P_{\g}} \subseteq T$ and $|T| = \ell$, and denote $S = [n]\setminus T$. 

Our goal is to find a subcode of $C$ of sufficient size, such that any $(k-2)$ codewords $x_1, x_2, \dots, x_{k-2}$ from this subcode collide in all the coordinates from $T$. In other words, for any coordinate $t \in T$ there should exist $i,j$ such that $(x_i)_t = (x_j)_t$. This will ensure that the coordinates from $T$ contribute $0$ to the RHS of \eqref{EHansel-code}, which will allow us to prove a better bound on the rate of the code $C$. We will now define the subcodes which satisfy this property.

First, denote by $\binom{\Sigma}{p}$ the family of $p$-element subsets of the alphabet $\Sigma = \{1, 2, \dots, k\}$. Then define
\[ \Omega = \underbrace{\binom{\Sigma}{k-3} \times \binom{\Sigma}{k-3} \times\dots\times \binom{\Sigma}{k-3}}_{\ell}. \]

Now, for any $\omega \in \Omega$ and any string $s \in \Sigma^{\ell}$, denote $s \vdash \omega$ if $s_1 \in \omega_1, s_2 \in \omega_2, \dots, s_{\ell} \in \omega_{\ell}$. Then, for any $\omega \in \Omega$, we define:
\[ C_{\omega} = \{x\in C\,:\,x_{\{T\}}\vdash \omega\},\]
where $x_{\{T\}}$ is the projection of the codeword $x$ on the set of coordinates $T$. Notice that $C_{\omega}$ has the property we discussed above. Indeed, for any pick $x_1, x_2, \dots, x_{k-2}\in C_{\omega}$ and any $t\in T$, it holds $(x_1)_t, (x_2)_t, \dots, (x_{k-2})_t \in \omega_t$, but $|\omega_t| = k-3$, and therefore $(x_1)_t, (x_2)_t, \dots, (x_{k-2})_t$ are not all distinct.

Denote then $M_{\omega} = |C_{\omega}|$. Note that for each $x \in C$ there are exactly $\binom{k-1}{k-4}^{\ell}$ different elements $\omega \in \Omega$ such that $x_{\{T\}}\vdash \omega$. Therefore
\[ \sum_{\omega \in \Omega} M_{\omega} = |C|\cdot \binom{k-1}{k-4}^{\ell}.\]

It suffices to prove that there exists at least one $\omega \in \Omega$ such that $M_{\omega} \geq n$ for our arguments further. For the sake of contradiction, suppose then that $M_{\omega} < n$ for all $\omega \in \Omega$. But then it holds
\[ 2^{nR} = |C| = \sum_{\omega \in \Omega} M_{\omega} \dfrac1{\binom{k-1}{k-4}^{\ell}} < \dfrac{\binom{k}{k-3}^{\ell}}{\binom{k-1}{k-4}^{\ell}}\cdot n = \left(\dfrac{k}{k-3}\right)^{\ell}n = 2^{\ell \cdot\log\frac{k}{k-3} + \log n} \leq 2^{nR},\]
where recall $\ell = \left\lfloor\frac{nR - \log n}{\log\left(\frac{k}{k-3}\right)} \right\rfloor$. Since we obtained a contradiction above, there exists $\omega \in \Omega$ such that $M_{\omega} \geq n$.

We are finally ready to describe the strategy to pick the codewords  $x_1, x_2, \dots, x_{k-2}$ in the almost balanced case. We do the following: first, deterministically choose some $\omega \in \Omega$ such that $M_{\omega} \geq n$, and then pick $x_1, x_2, \dots, x_{k-2}$ uniformly at random (without replacement) from $C_{\omega}$. Since all the codewords collide on the coordinates from the set $T$, we obtain in \eqref{EHansel-code}:

\begin{equation}
\label{Hansel_bal} \log(|C| - k + 2) \leq \sum_{m\in [n]}\E[\tau_m(x_1, x_2, \dots, x_{k-2})] = \sum_{m\in S}\E[\tau_m(x_1, x_2, \dots, x_{k-2})].
\end{equation}

Now fix some $m \in S$, and let $f_{m|\omega}$ be the frequency vector of the $m^{\text{th}}$ coordinate in the subcode $C_{\omega}$. 
Taking expectation over the choice of $x_1, x_2, \dots, x_{k-2}$ in \eqref{tau_precise} with respect to the the random strategy described above, we have
\begin{equation}
\begin{aligned}
 \E[\tau_m&(x_1, x_2, \dots, x_{k-2})]  \\
& =\dfrac{|C|}{|C|-k + 2}\prod_{j=0}^{k-3}\dfrac{|C_{\omega}|}{|C_{\omega}|-j}
\sum_{\substack{a_1,a_2,\dots, a_{k-2}\ \in\ \Sigma\\ \{a_s\} \text{ distinct}}} \left(1 - \sum_{s=1}^{k-2}f_m[a_s] \right)\cdot f_{m|\omega}[a_1]f_{m|\omega}[a_2]\dots f_{m|\omega}[a_{k-2}], 
\end{aligned}
\end{equation}
where the coefficients $\frac{|C_{\omega}|}{|C_{\omega}| - j}$, $j = 0, 1, \dots, (k-3)$ appear because we pick $(k-2)$ elements from $C_{\omega}$ without replacement. Since we took $\omega$ such that $|C_{\omega}| \geq n$, it holds $\frac{|C_{\omega}|}{|C_{\omega}| - j} \leq \frac{n}{n-j}$.

Recall that we defined the following function, which operates on probability vectors $f, g\in \R^k$:
\begin{equation}
\label{def-phi2} \phi_k(g, f) = \sum_{\substack{a_1,a_2,\dots, a_{k-2}\in\Sigma\\ \{a_s\} \text{ distinct}}}\ \prod_{s=1}^{k-2}g[a_s]\bigg(1 - \sum_{s=1}^{k-2}f[a_s]\bigg). 
\end{equation}
\noindent So we derive
\begin{equation}
\begin{aligned}
\label{w_phi}
 \E[\tau_m(x_1, x_2, \dots, x_{k-2})]  \leq \dfrac{n}{n-k+2}\prod_{j=0}^{k-3}\left(\dfrac{n}{n-j}\right) \phi_k(f_{m|\omega}, f_m) =  \phi_k(f_{m|\omega}, f_m)\cdot\big(1 + o(1)\big).
\end{aligned}
\end{equation}
\noindent
Consider the following definition:
\begin{equation}
\label{theta_g}
\boxed{\theta_k(\g) = \max_{g, f} \{\phi_k(g, f) : f,g\in \R^k \text{ are probability vectors, } \min\limits_{a\in \Sigma} f[a] \geq \g\}.}
\end{equation}
Let's first consider what bound we obtain using this definition, and then analyze how $\theta_k(\g)$ behaves.

Since $\min_{a\in \Sigma}f_m[a] \geq \g$ by construction of the set $S$, it holds $\phi_k(f_{m|\omega}, f_m) \leq \theta_k(\g)$ for any $m\in S$, so substituting it into \eqref{w_phi} gives	

\[ \E[\tau_m(x_1, x_2, \dots, x_{k-2})]  \leq  \theta_k(\g) \cdot\big(1 + o(1)\big).\]
Therefore, in \eqref{Hansel_bal} we derive
\begin{equation*}
\begin{aligned}
 \log(|C|-k+2) &\leq |S|\cdot\theta_k(\g)\big(1+o(1)\big) \\
 &= (n-\ell)\cdot\theta_k(\g)\big(1+o(1)\big) \\
 &\leq \left(n - \frac{nR}{\log \left(\frac{k}{k-3}\right)} + \frac{\log n}{\log \left(\frac{k}{k-3}\right)} + 1\right)\theta_k(\g)\big(1+o(1)\big).
 \end{aligned}
 \end{equation*}
Recall that $|C| = 2^{nR}$, thus for $n\to\infty$ we have

\[   R \leq \left(1 - \frac{R}{\log \left(\frac{k}{k-3}\right)}\right)\theta_k(\g) + o(1), \]
\begin{equation}
\label{R_bal}  
\boxed{R^{\text{bal}}_k(\g)\leq \dfrac{\theta_k(\g)}{1 + \frac{\theta_k(\g)}{\log \left(\frac{k}{k-3}\right)}} .}
\end{equation}

\medskip
It now remains to understand how $\theta_k(\g)$, defined in \eqref{theta_g}, behaves as a function of $\g$. 

\medskip \noindent \textbf{Continuity of $\theta_k$.} First of all, notice that for $\g = \frac1k$ there exists only one probability vector $f$ such that  $\min_{a\in \Sigma}f[a] \geq \g$, namely the uniform vector $u$. But then $\phi_k(g, u)$ is just an elementary symmetric sum of all the coordinates of $g$, and therefore, using the simple averaging argument similar to the one we used in Section \ref{unbalanced_sec}, we obtain $\phi_k(g, u)\leq \phi_k(u, u) = \alpha_k$, and so $\theta_k(1/k) = \alpha_k$.

Now we use an analysis instrument to prove that $\theta_k(\g)$ is continuous at $\g = \frac1k$. Namely, we use Berge's maximum theorem~\cite{Efe}, which we state here for completeness:
\begin{thm}[Berge's maximum theorem]
Let $X$ and $\Gamma$ be metric spaces, and $h\,:\,X \times \Gamma \to \R$ be a function jointly continuous in its two arguments. Let also $D\ :\ \Gamma \twoheadrightarrow X$ be a compact-valued correspondence, i.e. $D$ maps $\Gamma$ to the compact subsets of $X$: $D(\g) \subseteq X$ and $D(\g)$ is a compact for any $\g \in \Gamma$. Define for $\g\in\Gamma$:
\[ h^*(\g) = \max_{x \in D(\g)} h(x, \g).\]
If $D$ is continuous (both upper and lower hemicontinuous) at $\g$, then $h^*$ is continuous at $\g$.
\end{thm}
\noindent (The correspondences, or multi-valued functions, is the generalization of functions, and the continuity of correspondences is a generalization of continuity of functions. Refer to \cite{Efe}, Chapter E for details.)

\begin{claim}
\label{claim-cont}
$\theta_k(\g)$ is continuous at $\g = \frac1k$.
\end{claim}
\begin{proof}
This is a direct application of Berge's maximum theorem. Take in the settings of the theorem $X = \R^k\times \R^k$, $\Gamma = \R$, and function $h(x, \g) = h \big((g,f), \g\big) = \phi_k(g, f)$, where the variable $x$ is $(g, f) \in  \R^k\times \R^k$, and $h$ doesn't depend on the variable $\g$. Let us define the correspondence
 \[D(\g) = \Delta_k\times \Delta_k^{(\g)} \subseteq X,\] 
 where  
 \[\Delta_k = \{x\in \R_+^k\ :\ \langle x, \mathbf{1} \rangle = 1\}\] is the $k$-dimensional simplex, and
  \[\Delta_k^{(\g)} = \Big\{x\in \R_+^k\ :\ {\langle x, \mathbf{1} \rangle = 1};\quad x_i \geq \min\{\g, 1/k\}, \ \forall i \in [k]\footnote{We require this inequality rather than just $x \geq \g$ because in the latter case it would hold $D(1/k + \varepsilon) = \emptyset$ for any $\varepsilon > 0$, while $D(1/k)$ is a non-empty compact. In such settings $D(\g)$ wouldn't be continuous in $\g = 1/k$ for the natural reasons (though it would be ``one-side continuous'' in the settings of correspondences). Since we require $x\geq \min\{\g, 1/k\}$ component-wise, it will hold $D(1/k + \varepsilon) = D(1/k)$, and the correspondence $D$ will be continuous.} \Big\}.\] 
  It is clear that $D(\g)$ is a compact-valued correspondence, and it easily follows from the definition of continuity of correspondences that $D(\g)$ is continuous at any $\g$. By definition \eqref{theta_g}, for $\g \leq 1/k$ we have \[\theta_k(\g) = \max_{g, f} \{\phi_k(g, f) : (g,f)\in \Delta_k\times \Delta_k^{(\g)}\} = \max_{x \in D(\g)} h(x, \g) = h^*(\g)\] in the settings of the theorem. Therefore, we conclude that $\theta_k(\g)$ is continuous at any $\g$, in particular at ${\g = 1/k}$.
\end{proof}

\subsection{Improvement of the Fredman-Koml{\'o}s bound}
Now we prove that it is possible to choose such a threshold $\g$ that both bounds \eqref{k-R_unbal} and \eqref{R_bal} are stronger than the Fredman-Koml{\'o}s bound. Since for any code $C$ either the unbalanced or the almost balanced case holds, it will imply the general upper bound on $R_k$.

For $\g = \frac1k$, in \eqref{R_bal} we have
\begin{equation}
R^{\text{bal}}_k\left(\frac1k\right) \leq \dfrac{\theta_k(1/k)}{1 + \frac{\theta_k(1/k)}{\log \left(\frac{k}{k-3}\right)}} = \dfrac{\alpha_k}{1 + \frac{\alpha_k}{\log \left(\frac{k}{k-3}\right)}} < \alpha_k.
\end{equation}
Clearly, for small enough $\rho > 0$ it holds $\frac{\alpha_k + \rho}{1 + \frac{\alpha_k + \rho}{\log \left(\frac{k}{k-3}\right)}} < \alpha_k$. We proved in Claim~\ref{claim-cont} that $\theta_k(\g)$ is a continuous function, and we also have $\theta_k(1/k) = \alpha_k$. Therefore, there exists such small $\varepsilon > 0$	 that $\theta_k(1/k - \varepsilon) < \theta_k(1/k) + \rho = \alpha_k + \rho$, so
\[  R^{\text{bal}}_k\left(\frac1k - \varepsilon\right) \leq \frac{\alpha_k + \rho}{1 + \frac{\alpha_k + \rho}{\log \left(\frac{k}{k-3}\right)}} < \alpha_k.\]
Taking the same $\varepsilon$ in \eqref{cont_unbal}, we also have
\[  R^{\text{unbal}}_k\left(\frac1k - \varepsilon\right) < \alpha_k, \]
and therefore we obtain for the general bound
\[ \boxed{R_k \leq \max \bigg\{R^{\text{bal}}_k\left(1/k - \varepsilon\right);\ R^{\text{unbal}}_k\left(1/k - \varepsilon\right) \bigg\} < \alpha_k.}\]
This shows that we improved the Fredman-Koml{\'o}s bound for every $k$, at least by some tiny amount.

The above proof doesn't give an effective method to find an explicit new bound on $R_k$, since we use the continuity arguments in it. However, in the next section we present quite a different way to address the balanced case, and provide a specific way to obtain explicit bounds on $R_k$. Unfortunately, this approach leads to parametrized optimization of a large (increasing with $k$) number of polynomials. Using this method, we present new explicit bounds on $R_5$ and $R_6$ in section \ref{explicit_bounds}. Finally, we present a conjecture, assuming which we are able to find effective upper bound on $R_k$ stronger than the Fredman-Koml{\'o}s bound, via finding the root of a polynomial of degree $O(k)$ which lies within a specific interval.

\section{Optimization approach to the almost balanced case}
In this section we introduce the approach to estimate the value of $\theta_k(\g)$. Take $f$ to be a probability vector in $\R^k$, such that $f[a] \geq \g$ for $a \in \Sigma$. For convenience, we will use ``$f_a$'' to refer to the $a^{\text{th}}$ coordinate of $f$ rather then ``$f[a]$''. Recall:
\[ \phi_k(g, f) = \sum_{\substack{a_1,a_2,\dots, a_{k-2}\in\Sigma\\ \{a_i\} \text{ distinct}}}\ \prod_{i=1}^{k-2}g_{a_i}\bigg(1 - \sum_{i=1}^{k-2}f_{a_i}\bigg) = (k-2)!\Bigg[ \sum_{a_1 < a_2 < \dots < a_{k-2}\in\Sigma}\ \prod_{i=1}^{k-2}g_{a_i}\bigg(1 - \sum_{i=1}^{k-2}f_{a_i}\bigg)\Bigg].\]
Let $t = k(k-1)/2$ and let $P_1,P_2,\dots,P_t$ be an enumeration of all $(k-2)$-element subsets of $\Sigma=\{1,2,\dots,k\}$.

\noindent
Then we can rewrite
\[  \phi_k(g, f) = (k-2)!\Bigg[ \sum_{j=1}^t\prod_{a \in P_j}g_a\bigg(1 - \sum_{a \in P_j}f_a\bigg)\Bigg].\]
Let $d_j = \prod_{a \in P_j}g_a$, $w_j = \bigg(1 - \sum_{a \in P_j}f_a\bigg)$ for $j = 1, 2, \dots, t$, and therefore we simply have
\[ \phi_k(g, f) = (k-2)!\sum_{i=1}^td_iw_i. \]
Since $f_a \geq \g$ for any $a\in \Sigma$, it holds $w_i \leq 1 - (k-2)\g$. On the other hand, since $f$ is a probability vector, $w_i = f_{a_1} + f_{a_2}$, where $\{a_1, a_2\} = \Sigma\setminus P_i$, and so $w_i \geq 2\g$. Denoting $w_i' = w_i - 2\g$, we obtain: \[0 \leq w_i' \leq (1-k\g),\]
\[ \phi_k(g, f) = (k-2)!\bigg[2\g\left(\sum_{i=1}^td_i\right) + \sum_{i=1}^td_iw_i' \bigg].\]
Since $w_i = f_{a_1} + f_{a_2}$ for $\{a_1, a_2\} = \Sigma\setminus P_i$, we can argue that $\sum_{i=1}^tw_i$ is symmetric sum of all $f_a$ for $a\in \Sigma$, where each $f_a$ occurs $\binom{k-1}{k-2} = (k-1)$ times. Since $\sum_{a\in \Sigma}f_a = 1$, it then holds
\[\sum_{i=1}^tw_i = (k-1)\sum_{a\in\Sigma}f_a = (k-1) \quad\Rightarrow\quad \sum_{i=1}^tw_i' = (k-1) - 2\g t = (k-1)\big(1 - k\g\big).\]

Then consider the following optimization problem: 
\begin{align*} \max_y\quad &\sum_{i=1}^{t}d_iy_i,\\
 \text{s.t.}\quad & 0\leq y_i \leq (1-k\g),\\
 &\sum_{i=1}^{t}y_i = (k-1)\big(1 - k\g\big).
\end{align*}
Note that the vector $(w_1', w_2',\dots, w_t')$ is feasible for the above program, and let $y^*$ be the optimal solution for this program. Then we have:
\begin{equation}
\label{k-phi_d}
\phi_k(g,f) = (k-2)!\bigg[2\g\left(\sum_{i=1}^{t} d_i\right) + \sum_{i=1}^{t} d_iw_i'\bigg] \leq (k-2)!\bigg[2\g\left(\sum_{i=1}^{t} d_i\right) + \sum_{i=1}^{t} d_iy^*_i\bigg].
\end{equation}

It is straightforward to see that the optimal solution $y^*$ to the above program has $(k-1)$ non-zero coordinates, corresponding to the first $(k-1)$ greatest values among $\{d_1, d_2, \dots, d_{t}\}$, each equal to $(1-k\g)$, and zeros in the remaining coordinates. In other words, denote $d_{(i)}$ to be the $i^{\text{th}}$ ordered statistic of the set $\{d_1, d_2, \dots, d_{t}\}$, so $d_{(1)}$ is the maximum of this set, and $d_{(t)}$ is the minimum. Then
\[ \sum_{i=1}^{t} d_iy^*_i = (1-k\g)\sum_{j=1}^{k-1}d_{(j)},\]
and therefore in \eqref{k-phi_d} obtain
\begin{equation}
\label{k-stat_bound}
\phi_k(g,f) \leq (k-2)!\Bigg[(1-(k-2)\g)\left(\sum_{j=1}^{k-1} d_{(j)}\right) + 2\g\sum_{j=k}^{t} d_{(j)}\Bigg].
\end{equation}

Without loss of generality, assume $g_1\geq g_2 \geq \dots \geq g_k \geq 0$. Even with this fixed ordering of the coordinates of $g$, there still could be different cases of ordering within the set $\{d_1, d_2, \dots, d_{t}\}$, and for each of these cases we obtain some different function of $d_i$'s in the RHS of \eqref{k-stat_bound}. Say there can be $q_k$ different possible sets of first $(k-1)$ ordered statistics within the set $\{d_1, d_2, \dots, d_{t}\}$, then there would be $q_k$ different functionals of $d_i$'s, and thus of $g_i$'s, in the RHS of \eqref{k-stat_bound}, call them $\Theta_{k}^{(1)}(g, \g), \Theta_{k}^{(2)}(g, \g),\dots, \Theta_{k}^{(q_k)}(g, \g)$. Since exactly one ordering is correct for any particular vector $g$, we obtain
\[ \phi_k(g,f)  \leq \max_{i = 1, 2 ,\dots, q_k} \Theta_{k}^{(i)}(g, \g).\]
Then define
\begin{equation}
\label{max_problems} \theta_{k}^{(i)}(\g) = \max_x\bigg\{\Theta_{k}^{(i)}(x, \g) : \sum_{j=1}^kx_j= 1, x \geq 0\bigg\}, \qquad \text{for } i = 1, 2, \dots, q_k,
\end{equation}
and so the quantity $\theta_k(\g)$ defined in \eqref{theta_g} satisfies
\[ \theta_k(\g) \leq \max_{i = 1, 2 ,\dots, q_k} \theta_{k}^{(i)}(\g).\]
So to find an upper bound on $\theta_k(\g)$ it suffices to find the maximums of all $\theta_{k}^{(i)}(\g)$ for $i = 1, 2, \dots, q_k$. However, $q_k$ grows rapidly as $k$ increases, so it is not clear how to do this efficiently. We introduce a conjecture below, which suggests that we can determine which of the values $\theta_{k}^{(i)}(\g)$, $i=1,2,\dots, q_k$, is the greatest for any $k$.

Specifically, the conjecture is stated as follows: we assume that the maximum among all the values $\theta_{k}^{(i)}(\g)$, $i = 1, 2, \dots, q_k$, is the greatest for the functional $\Theta_{k}^{(i)}(x, \g)$ corresponding to the case, when the first $(k-1)$ ordered statistics of the set $\{d_1, d_2, \dots, d_{t}\}$ form the set $\left\{\frac{\prod_{i=1}^{k-1}g_i}{g_a}\right\}_{a\in [k-1]}$. In other words, $\{d_{(1)}, d_{(2)}, \dots, d_{(k-1)}\}$ correspond to the sets $P_j$ that contain all their $(k-2)$ elements from $\{1,2,\dots,k-1\}$ (recall $d_j = \prod_{a \in P_j}g_j$). So the first $(k-1)$ ordered statistics are formed as the products of only the first $(k-1)$ coordinates of $g$, ignoring the coordinate $g_k$.

Recall that we denote by $S_h^t(g)$ the $h$-th elementary symmetric sum of a set of $t$ coordinates $\{g_1, g_2, \dots, g_t\}$ (i.e. the sum of all products of $h$ distinct elements from the first $t$ coordinates of $g$). Then the above-mentioned conjecture can be formalized as follows:
\begin{conj}
\label{max_conj}
\begin{equation}
\label{k-theta}
 \theta_k(\g) = \max_x\bigg\{(k-2)!\Bigg[\bigg(1-(k-2)\g\bigg)S_{k-2}^{k-1}(x) + 2\g\cdot x_{k}\cdot S_{k-3}^{k-1}(x)\Bigg]:\quad \sum_{i=1}^kx_i= 1, x \geq 0\bigg\}.
 \end{equation}
\end{conj}

\medskip
\noindent
Indeed, the function 
\[ \Theta_{k}^{\g}(g) = (k-2)!\Bigg[\bigg(1-(k-2)\g\bigg)S_{k-2}^{k-1}(g) + 2\g\cdot g_{k}\cdot S_{k-3}^{k-1}(g)\Bigg] \]
just corresponds to the functional in the RHS of \eqref{k-stat_bound} in the case we discussed above, since in this case the elements $\{d_{(1)}, d_{(2)}, \dots, d_{(k-1)}\}$ are just the summands from $S_{k-2}^{k-1}$, while all the other elements $\{d_{(k)}, d_{(k+1)}, \dots, d_{(t)}\}$ are of type $g_{k}\cdot\prod_{i=1}^{k-3}g_{a_i}$.
 
Let's now find the RHS of \eqref{k-theta}. Since $1-(k-2)\g \geq 2\g$, it is easy to see that we may consider only vectors $x$ for which $x_k$ is the minimal over all other coordinates. Indeed, in other case, if $x_k > x_i$, switching the values in $x_k$ and $x_i$ will only increase the functional $\Theta_k^{\g}(x)$. So from now on, we consider $x_k$ to be minimal among the coordinates of $x$.

Now notice that the sums $S_{k-2}^{k-1}(x)$ and $S_{k-3}^{k-1}(x)$ are elementary symmetric sums with respect to $x_1, x_2, \dots, x_{k-1}$, which means that if some two of $x_i, x_j$ are different for $i,j \leq k-2$, then we can substitute them both by their average, and the functional $\Theta_k(x)$ will not decrease. Therefore, we conclude that the maximum of RHS of \eqref{k-theta} is achieved when $x_1 = x_2 = \dots = x_{k-1} = \beta$, and then $x_{k} = 1 - (k-1)\beta$, so it must hold $\beta \geq 1 - (k-1)\beta \geq 0$, thus $\frac1k \leq \beta \leq \frac1{k-1}$. Denoting $Q_k^{\g}(\beta) = \Theta_k^{\g}(\beta, \beta, \dots, 1 - (k-1)\beta)$, we obtain (assuming Conjecture~\ref{max_conj}):
\[\theta_k(\g) = \max_{\frac1{k} \leq \beta \leq \frac1{k-1}} Q_k^{\g}(\beta).\]
Then we compute:
\begin{equation*}
\begin{aligned}
 Q_k^{\g}(\beta) &= (k-2)!\Bigg[\bigg(1-(k-2)\g\bigg)(k-1)\beta^{k-2} + 2\g \bigg(1 - (k-1)\beta\bigg)\frac{(k-1)(k-2)}{2}\beta^{k-3}\Bigg] \\
 &= (k-1)!\beta^{k-3}\left( \beta\bigg(1 - (k^2-2k)\g\bigg) + (k-2)\g \right);\\
 \left(Q_k^{\g}(\beta)\right)' &=  (k-1)!(k-2)\beta^{k-4}\left( \beta\bigg(1 - (k^2-2k)\g\bigg) + (k-3)\g \right);\\
 \left(Q_k^{\g}(\beta)\right)'' &=  (k-1)!(k-2)(k-3)\beta^{k-5}\left( \beta\bigg(1 - (k^2-2k)\g\bigg) + (k-4)\g \right). 
\end{aligned}
\end{equation*}
The second derivative of $Q_k^{\g}(\beta)$ is negative whenever $\beta > \frac{(k-4)\g}{(k^2-2k)\g - 1}$, and it is easy to see that $\frac{(k-4)\g}{(k^2-2k)\g - 1} < \frac1k$ if $\g > \frac1{2k}$. Since we initially considered $\g\geq\frac1{2k-3} > \frac1{2k}$, it holds $\frac{(k-4)\g}{(k^2-2k)\g - 1} < \frac1k \leq \beta \leq \frac1{k-1}$. So $\left(Q_k^{\g}(\beta)\right)'' < 0$ for $\beta$ being in the interval of interest.

\noindent Next, since $\g\leq \frac1k$, it is straightforward to see that $\left(Q_k^{\g}\left(\frac1k\right)\right)' \geq 0$.

\noindent Finally, for $\g\geq\frac1{2k-3}$ it is easy to check that $\left(Q_k^{\g}\left(\frac1{k-1}\right)\right)' \leq 0$.

Altogether, we can conclude that for the interval $\frac1{2k-3} \leq \g \leq \frac1k$ the maximum of $Q_k^{\g}(\beta)$ can be found by solving the equation $\left(Q_k^{\g}(\beta)\right)' = 0$, and thus the optimal solution is $\beta^* = \dfrac{(k-3)\g}{(k^2-2k)\g - 1}$. Finally:
\begin{equation}
\label{k-theta-final}
\boxed{\theta_k(\g) = Q_k^{\g}(\beta^*) = \dfrac{(k-1)!(k-3)^{k-3}\g^{k-2}}{\big((k^2-2k)\g - 1\big)^{k-3}},}
\end{equation}
assuming the conjecture \ref{max_conj} holds.

\subsection{Optimal threshold (assuming the conjecture)}
The optimal threshold $\g$ is such that the bounds \eqref{k-R_unbal} and \eqref{R_bal} are equal, since we might only claim $R_k \leq \max\{R^{\text{bal}}_k, R^{\text{unbal}}_k\}$, and the first bound becomes weaker as $\g$ increases, while the second bound becomes stronger. Therefore, the optimal threshold is the solution of the following equation:
\begin{equation}
\label{balanced_gamma} R^{\text{unbal}}_k\left(\g\right) = \dfrac{\theta_k(\g)}{1 + \frac{\theta_k(\g)}{\log \left(\frac{k}{k-3}\right)}} = \dfrac{\alpha_k}{1 + \frac{\alpha_k - \xi_k(\g)}{\log \left(\frac{k}{k-3}\right)}} = R^{\text{bal}}_k\left(\g\right),
\end{equation}
where $\alpha_k = \dfrac{k!}{k^{k-1}}$ is the Fredman-Koml{\'o}s bound, $\theta_k(\g)$ can be found using expression \eqref{k-theta-final} (assuming the conjecture holds), and $\xi_k(\g)$ is found via \eqref{def_xi_k}. Note that both $\theta_k(\g)$ and $\xi_k(\g)$ are rational functions with degrees $O(k)$, and therefore the above equation is equivalent to finding a root of a polynomial of degree $O(k)$ in $\g$, that lies in the interval $\left(\frac{1}{2k-3}, \frac1k\right)$. Such a solution certainly exists, because $ R^{\text{bal}}_k(1/k) = \alpha_k > R^{\text{unbal}}_k(1/k)$, but it is easy to verify that ${R^{\text{bal}}_k\left(\frac1{2k-3}\right) < \alpha_k < R^{\text{unbal}}_k\left(\frac1{2k-3}\right)}$. Therefore, there exists a point $\g^*\in\left(\frac{1}{2k-3}, \frac1k\right)$ where $R^{\text{unbal}}_k\left(\g^*\right) = R^{\text{bal}}_k\left(\g^*\right)$, since these functions are continuous.

We can now note that much weaker version of the conjecture would be enough to be able to use these arguments. In fact, we only need for the equation \eqref{k-theta} to hold specifically for the value of the threshold $\g = \g^*$. So it just suffices to solve all optimization problems \eqref{max_problems} for this value $\g^*$, and check if the conjecture indeed holds (namely, that the maximum of $\Theta_k^{(i)}(g,\g)$ is the greatest for the functional $\Theta_k^{\g}(g)$ described in the conjecture). 
In case it actually holds, we are able to apply inequalities \eqref{R_bal} and \eqref{k-R_unbal} and obtain an explicit upper bound on $R_k$, which beats the Fredman-Koml{\'o}s bound. In the next section we do so for $k=5$ and $k=6$.

\subsection{New bounds for $k=5$ and $k=6$}
\label{explicit_bounds}

Applying \eqref{balanced_gamma} for $k=5$, for instance, is just solving the following equation:

\[\dfrac{96\g^3\log \frac52}{(15\g - 1)^2\log \frac52 + 96\g^3} = \frac{32\alpha_5\log \frac52}{32\log \frac52 + 32\alpha_5 - 3(1-\g)^3(15\g+1)},\]
where $\frac17\leq \g\leq \frac15$. The only feasible solution to the above is $\g^* \approx 0.136163$. We then use the tools of numerical optimization in {\it Wolfram Mathematica} \cite{Mathematica} to verify that the conjecture holds. For this case, we have $q_5 = 2$, so we only need to optimize two functionals over the simplex, and compare the two optimal values. After verifying the conjecture, we obtain the new bound for $5$-hashing: \[R_5 <  0.190825 < 0.192 = \dfrac{24}{125} = \alpha_5.\]

For $k=6$, the above approach gives us:
\[ R_6 < 0.0922787 < 0.0\overline{925} = \dfrac{5}{54} = \alpha_6.\]
($0.0\overline{925}$ stands for repeating decimal $0.0925925925\dots$.)

\appendix
\section{$(b,k)$-hashing}
\label{(b,k)-sec}
As we mentioned in the Introduction, the problem for which Fredman and Koml{\'o}s \cite{Fredman} proved a bound was in fact broader then the $k$-hashing problem. Namely, for $b \ge k$, say that a code $C \subseteq [b]^n$ is a $(b,k)$-hash code if for any $k$ distinct codewords from $C$ there exists a coordinate in which all these codewords differ. Then the $(b,k)$-hashing problem consists in estimating the maximum possible rate $R_{(b,k)}$ of $(b,k)$-hash codes. This can be equivalently formulated in the context of hash functions. 

All the bounds for this generalized version of the problem rely on extended versions of the Hansel lemma. Fredman and Koml{\'o}s \cite{Fredman} allowed for the graphs $G_i$ in the settings of Lemma \ref{Hansel-lemma} to be multipartite rather then just bipartite, and later K{\"o}rner and Marton \cite{Korner} also proved the generalization of the lemma for hypergraphs. The generalized version of the lemma was also proven in \cite{Nilli} using probabilistic arguments. 

\begin{lem}[Hansel for hypergraphs]
\label{hyper-Hansel-lemma}
Let $K^{(d)}_m$ be a complete $d$-uniform hypergraph on $m$ vertices. Let also $G_1, G_2, \dots, G_t$ be $c$-partite $d$-uniform hypergraphs, such that $E\left(K^{(d)}_m\right) = \bigcup\limits_{i=1}^tE(G_i)$. Denote by $\tau(G_i)$ the fraction of non-isolated vertices in $G_i$. Then the following holds:
\begin{equation}
\label{hyper-Hansel}
\dfrac{\log m}{\log (d-1)} \leq \log \dfrac{c}{d-1} \cdot \sum_{i=1}^t\tau(G_i).
\end{equation}
\end{lem}

Again, to get the bound on the rates of $(b,k)$-codes, consider some $(b,k)$-hash code $C\subseteq [b]^n$. Take a subset of this code $\{x_1, x_2, \dots, x_j\}\subseteq C$, where $1\leq j \leq k-2$. We now define $(b-j)$-partite $(k-j)$-uniform hypergraphs $G_i^{x_1, \dots, x_j}$, for $i \in [n]$, as follows: 
\begin{equation*}
\begin{aligned}
& V(G_i^{x_1, \dots, x_j}) = C\setminus \{x_1, x_2, \dots, x_j\},\\
& E(G_i^{x_1, \dots, x_j}) = \Bigg\{ \{y_1, y_2, \dots, y_{k-j}\}\ :\ (y_{1})_i, (y_2)_i, \dots, (y_{k-j})_i, (x_1)_i, (x_2)_i, \dots, (x_j)_i \text{ are distinct} \Bigg\}.
\end{aligned}
\end{equation*}

Directly applying the the above Hansel lemma for hypergraphs and denoting $\tau_i(x_1, x_2, \dots, x_j) = \tau\left(G_i^{x_1, \dots, x_j}\right)$, we obtain:
\begin{equation}
\label{hyper-Hansel-code}
\dfrac{\log \left(|C| - j\right)}{\log (k - j - 1)} \leq \log \dfrac{b - j}{k - j - 1} \sum_{i=1}^n\tau_i(x_1, x_2, \dots, x_j).
\end{equation}     
Similarly, one then might use different ways to pick $x_1, x_2, \dots, x_j$ in order to obtain the upper bound on the rate of $C$ from the above.

 In \cite{Fredman} for the usual graph case ($j = k-2$), and then in \cite{Korner} for hypergraphs, the codewords $x_1, x_2, \dots, x_j$ are picked independently at random from the code $C$, and \eqref{hyper-Hansel-code} gives the following bound (K{\"o}rner-Marton bound):
	\begin{equation}
	\label{Korner-Marton} R_{(b,k)} \leq \min_{0\leq j\leq k-2} \dfrac{b^{\underline{j+1}}}{b^{j+1}}\log \dfrac{b-j}{k-j-1},
	\end{equation}
	where $b^{\underline{j+1}} = b(b-1)\dots(b-j)$.

Note that for the case $b = k$ ($k$-hashing) it can be shown that the above minimum is attained at $j = k-2$. But in this case the bound \eqref{Korner-Marton} turns into the Fredman-Koml{\'o}s bound \eqref{Fredman-Kolmos}, so this approach doesn't give any improvement for $k$-hashing.

In \cite{Arikan} Arikan, using the rate versus distance ideas discussed in the section \ref{approach}, provides the following bound on the rate $R_{(b,k)}$ for general $b$ and $k$:
\begin{equation}
\label{Arikan}
\begin{aligned}
R_{(b,k)} \leq \sup_x \{x \leq \alpha_j(x), j = 2, \dots, k-2\},
\end{aligned}
\end{equation}
where 
\[ \alpha_j(x) = \dfrac{b-j}{k-1}2^{-x}\left(1 - \dfrac{x}{\log b}\right)\dfrac{b^{\underline{j}}}{b^j}\log \dfrac{b-j}{k-1-j}\]
for $j = 2, \dots, 	b-k$, and 
\[ \alpha_j(x) = \left(1 - \dfrac{j}{b - k + 1}\left(1 - 2^{-x}\right)\right)\left(1 - \dfrac{x}{\log b}\right)\dfrac{b^{\underline{j}}}{b^j}\log \dfrac{b-j}{k-1-j}\]
for $j = b-k+1, \dots, 	k-2$.

Arikan's bound improves the Fredman-Koml{\'o}s bound \eqref{Fredman-Kolmos} for $b=k=4$, and also beats the K\"{o}rner-Marton bound \eqref{Korner-Marton} for many pairs of $(b,k)$; see \cite{Arikan}. However, neither \eqref{Korner-Marton} nor \eqref{Arikan} beat the bound \eqref{Fredman-Kolmos} when $b=k > 4$.

The approach we described in this paper generalizes to the settings of $(b,k)$-hashing problem in a straightforward way, improving the K{\"o}rner-Marton bound \eqref{Korner-Marton} for any $j = 0, 1, \dots, k-2$. However, since this bound was already possibly beaten by the Arikan's bound \eqref{Arikan} for certain settings of $b\not= k$, and in order to keep the presentation simple, we don't include the proofs for $(b,k)$-hashing problem in this paper.

\bibliography{bib}{}

\begin{thebibliography}{10}
\providecommand{\url}[1]{#1}
\csname url@samestyle\endcsname
\providecommand{\newblock}{\relax}
\providecommand{\bibinfo}[2]{#2}
\providecommand{\BIBentrySTDinterwordspacing}{\spaceskip=0pt\relax}
\providecommand{\BIBentryALTinterwordstretchfactor}{4}
\providecommand{\BIBentryALTinterwordspacing}{\spaceskip=\fontdimen2\font plus
\BIBentryALTinterwordstretchfactor\fontdimen3\font minus
  \fontdimen4\font\relax}
\providecommand{\BIBforeignlanguage}[2]{{%
\expandafter\ifx\csname l@#1\endcsname\relax
\typeout{** WARNING: IEEEtran.bst: No hyphenation pattern has been}%
\typeout{** loaded for the language `#1'. Using the pattern for}%
\typeout{** the default language instead.}%
\else
\language=\csname l@#1\endcsname
\fi
#2}}
\providecommand{\BIBdecl}{\relax}
\BIBdecl

\bibitem{elias88}
\BIBentryALTinterwordspacing
P.~Elias, ``Zero error capacity under list decoding,'' \emph{{IEEE} Trans.
  Information Theory}, vol.~34, no.~5, pp. 1070--1074, 1988. [Online].
  Available: \url{https://doi.org/10.1109/18.21233}
\BIBentrySTDinterwordspacing

\bibitem{Dalai}
M.~Dalai, V.~Guruswami, and J.~Radhakrishnan, ``An improved bound on the
  zero-error list-decoding capacity of the 4/3 channel,'' in \emph{2017 IEEE
  International Symposium on Information Theory (ISIT)}, June 2017, pp.
  1658--1662.

\bibitem{Fredman}
\BIBentryALTinterwordspacing
M.~L. Fredman and J.~Koml{\'{o}}s, ``On the size of separating systems and
  families of perfect hash functions,'' \emph{{SIAM} Journal on Algebraic
  Discrete Methods}, vol.~5, no.~1, pp. 61--68, mar 1984. [Online]. Available:
  \url{https://doi.org/10.1137%2F0605009}
\BIBentrySTDinterwordspacing

\bibitem{Korner86}
\BIBentryALTinterwordspacing
J.~K\"{o}rner, ``Fredman\textendash{K}oml{\'{o}}s bounds and information
  theory,'' \emph{{SIAM} Journal on Algebraic Discrete Methods}, vol.~7, no.~4,
  pp. 560--570, oct 1986. [Online]. Available:
  \url{https://doi.org/10.1137%2F0607062}
\BIBentrySTDinterwordspacing

\bibitem{Arikan}
\BIBentryALTinterwordspacing
E.~Arikan, ``A bound on the zero-error list coding capacity,'' in
  \emph{Proceedings. {IEEE} International Symposium on Information
  Theory}.\hskip 1em plus 0.5em minus 0.4em\relax {IEEE}. [Online]. Available:
  \url{https://doi.org/10.1109%2Fisit.1993.748467}
\BIBentrySTDinterwordspacing

\bibitem{Korner}
\BIBentryALTinterwordspacing
J.~K\"{o}rner and K.~Marton, ``New bounds for perfect hashing via information
  theory,'' \emph{European Journal of Combinatorics}, vol.~9, no.~6, pp.
  523--530, nov 1988. [Online]. Available:
  \url{https://doi.org/10.1016%2Fs0195-6698%2888%2980048-9}
\BIBentrySTDinterwordspacing

\bibitem{hansel}
G.~Hansel, ``Nombre minimal de contacts de fermature n\'{e}cessaires pour
  r\'{e}aliser une fonction bool\'{e}enne sym\'{e}trique de $n$ variables,''
  \emph{C. R. Acad. Sci. Paris}, pp. 6037--6040, 1964.

\bibitem{Kornerentropy}
J.~K{\"o}rner, ``Coding of an information source having ambiguous alphabet and
  the entropy of graphs,'' \emph{6th Prague Conference on Information Theory},
  pp. 411--425, 1973.

\bibitem{Radhakrishnan}
J.~Radhakrishnan, ``Entropy and counting,'' 2001.

\bibitem{Nilli}
\BIBentryALTinterwordspacing
A.~Nilli, ``Perfect hashing and probability,'' \emph{Combinatorics, Probability
  and Computing}, vol.~3, no.~03, pp. 407--409, sep 1994. [Online]. Available:
  \url{https://doi.org/10.1017%2Fs0963548300001280}
\BIBentrySTDinterwordspacing

\bibitem{Arikan2}
\BIBentryALTinterwordspacing
E.~Arikan, ``An upper bound on the zero-error list-coding capacity,''
  \emph{{IEEE} Transactions on Information Theory}, vol.~40, no.~4, pp.
  1237--1240, jul 1994. [Online]. Available:
  \url{https://doi.org/10.1109%2F18.335947}
\BIBentrySTDinterwordspacing

\bibitem{Efe}
E.~Ok, \emph{\BIBforeignlanguage{English (US)}{Real analysis with economic
  applications}}.\hskip 1em plus 0.5em minus 0.4em\relax Princeton University
  Press, 9 2011.

\bibitem{Mathematica}
W.~R. Inc., ``Mathematica, {V}ersion 11.3,'' champaign, IL, 2018.

\end{thebibliography}
\bibliographystyle{IEEEtran}
\end{document}